\documentclass[twoside,letter]{article} 

\usepackage[accepted]{aistats2017}
\usepackage[utf8]{inputenc}
\usepackage{mathtools,bm}
\usepackage{amsmath}
\usepackage{amssymb}
\usepackage{amsthm}
\usepackage{color}
\usepackage{subcaption}
\usepackage{enumitem}
\usepackage{url}

\usepackage[colorinlistoftodos]{todonotes}

\mathtoolsset{showonlyrefs=true}
\providecommand{\abs}[1]{\left\lvert#1\right\rvert} 
\providecommand{\norm}[1]{\left\lVert#1\right\rVert} 
\providecommand{\normp}[3]{\left\lVert#1\right\rVert_{#2,#3}} 

\DeclareMathOperator*{\argmin}{\arg\!\min}

\newcommand{\cB}{\mathcal{B}}

\newcommand{\cI}{\mathcal{I}}

\newcommand{\cL}{\mathcal{L}}

\newcommand{\cP}{\mathcal{P}}

\newcommand{\cT}{\mathcal{T}}

\newtheorem{definition}{Definition}
\newtheorem{theorem}{Theorem}
\newtheorem{lemma}{Lemma}

\newcommand{\unaryminus}{\textrm{-}}

\begin{document}

%
\runningtitle{Learning Nash Equilibrium for General-Sum Markov Games from Batch Data}

\runningauthor{Julien Pérolat, Florian Strub, Bilal Piot, Olivier Pietquin}

\twocolumn[
\aistatstitle{Learning Nash Equilibrium for General-Sum\\ Markov Games from Batch Data}

\aistatsauthor{ Julien Pérolat\textsuperscript{(1)} \\ \small julien.perolat@ed.univ-lille1.fr\normalsize  \And Florian Strub\textsuperscript{(1)} \\ \small florian.strub@ed.univ-lille1.fr\normalsize \And Bilal Piot\textsuperscript{(1,3)} \\ \small bilal.piot@univ-lille1.fr\normalsize \And Olivier Pietquin\textsuperscript{(1,2,3)} \\ \small olivier.pietquin@univ-lille1.fr\normalsize}

\aistatsaddress{ \textsuperscript{(1)}Univ. Lille, CNRS, Centrale Lille, Inria, UMR 9189 - CRIStAL, F-59000 Lille, France\\
\textsuperscript{(2)}Institut Universitaire de France (IUF), France \\
\textsuperscript{(3)}DeepMind, UK}]

\begin{abstract}
This paper addresses the problem of learning a Nash equilibrium in $\gamma$-discounted multiplayer general-sum Markov Games (MGs) in a batch setting. As the number of players increases in MG, the agents may either collaborate or team apart to increase their final rewards. One solution to address this problem is to look for a Nash equilibrium. Although, several techniques were found for the subcase of two-player zero-sum MGs, those techniques fail to find a Nash equilibrium in general-sum Markov Games. 
In this paper, we introduce a new definition of $\epsilon$-Nash equilibrium in MGs which grasps the strategy's quality for multiplayer games. We prove that minimizing the norm of two Bellman-like residuals implies to learn such an $\epsilon$-Nash equilibrium. Then, we show that minimizing an empirical estimate of the $\cL_p$ norm of these Bellman-like residuals allows learning for general-sum games within the batch setting. Finally, we introduce a neural network architecture that successfully learns a Nash equilibrium in generic multiplayer general-sum turn-based MGs.
\end{abstract}
\vspace{-0.8em}
\section{Introduction}
\vspace{-0.8em}
A Markov Game (MG) is a model for Multi-Agent Reinforcement Learning (MARL)~\cite{littman1994markov}. Chess, robotics or human-machine dialogues are a few examples of the myriad of applications that can be modeled by an MG. At each state of an MG, all agents have to pick an action simultaneously. As a result of this mutual action, the game moves to another state and each agent receives a reward. 
In this paper, we focus on presently the most generic MG framework with perfect information: general-sum multiplayer MGs. This framework includes MDPs (which are single agent MGs), zero-sum two-player games, normal form games, extensive form games and other derivatives~\cite{nisan2007algorithmic}. 
For instance, general-sum multiplayer MGs have neither a specific reward structure as opposed to zero-sum two-player MGs~\cite{shapley1953stochastic,perolat} nor state space constraints such as the tree structure in extensive form games~\cite{nisan2007algorithmic}. 
More importantly, when it comes to MARL, using classic models of reinforcement learning (MDP or Partially Observable MDP) assumes that the other players are part of the environment. Therefore, other players have a stationary strategy and do not update it according to the current player.
General-sum multiplayer MGs remove this assumption by allowing the agents to dynamically re-adapt their strategy according the other players. Therefore, they may either cooperate or to confront in different states which entails joint-strategies between the players.


In this paper, agents are not allowed to directly settle a common strategy beforehand. They have to choose their strategy secretly and independently to maximize their cumulative reward.
In game theory, this concept is traditionally addressed by searching for a Nash equilibrium~\cite{filar2012competitive}. Note that different Nash equilibria may co-exist in a single MG. 
Computing a Nash equilibrium in MGs requires a perfect knowledge of the dynamics and the rewards of the game. In most practical situations, the dynamics and the rewards are unknown. To overcome this difficulty two approaches are possible: 
Either learn the game by interacting with other agents or extract information from the game through a finite record of interactions between the agents. The former is known as the online scenario while the latter is called the batch scenario. This paper focuses on the batch scenario.
Furthermore, when a player learns a strategy, he faces a representation problem. In every state of the game (or configuration), he has to store a strategy corresponding to his goal. However, this approach does not scale with the number of states and it prevents from generalizing a strategy from one state to another. Thus, function approximation is introduced to estimate and generalize the strategy.



First, the wider family of batch algorithms applied to MDPs builds upon approximate value iteration~\cite{riedmiller2005neural,munos2008finite} or approximate policy iteration~\cite{lagoudakis2003reinforcement}. Another family relies on the direct minimization of the optimal Bellman residual~\cite{baird1995residual,Piot_DC_NIPS}. These techniques can handle value function approximation and have proven their efficacy in large MDPs especially when combined with neural networks~\cite{riedmiller2005neural}. The batch scenario is also well studied for the particular case of zero-sum two-player MGs. It can be solved with analogue techniques as the ones for MDPs~\cite{lagoudakis2002value,perolat,perolat_non-stat_AISTATS}. 
All the previously referenced algorithms use the value function or the state-action value function to compute the optimal strategy. However, \cite{zinkevich2006cyclic} demonstrates that no algorithm based on the value function or on the state-action value function can be applied to find a stationary-Nash equilibrium in general-sum MGs. Therefore, recent approaches to find a Nash equilibrium consider an optimization problem on the strategy of each player~\cite{prasad2015two} in addition to the value functions. Finding a Nash equilibrium is well studied when the model of the MG is known~\cite{prasad2015two} or estimated~\cite{akchurina2010multi} or in the online scenario~\cite{littman2001friend,prasad2015two}. However, none of the herein-before algorithms approximates the value functions or the strategies. To our knowledge, using function approximation to find a Nash-equilibrium in MGs has not been attempted before this paper.

The key contribution of this paper is to introduce a novel approach to learn an $\epsilon$-Nash equilibrium in a general-sum multiplayer MG in the batch setting. This contribution relies on two ideas: the minimization of two Bellman-like residuals and the definition of the notion of weak $\epsilon$-Nash equilibrium. Again, our approach is the first work that considers function approximation to find an $\epsilon$-Nash equilibrium. As a second contribution, we empirically evaluate our idea on randomly generated deterministic MGs by designing a neural network architecture. This network minimizes the Bellman-like residuals by approximating the strategy and the state-action value-function of each player. 

These contributions are structured as follows: first, we recall the definitions of an MG, of a Nash equilibrium and we define a weaker notion of an $\epsilon$-Nash equilibrium. Then, we show that controlling the sum of several Bellman residuals allows us to learn an $\epsilon$-Nash equilibrium. Later, we explain that controlling the empirical norm of those Bellman residuals addresses the batch scenario problem. At this point, we provide a description of the NashNetwork. Finally, we empirically evaluate our method on randomly generated MGs (Garnets). Using classic games would have limited the scope of the evaluation. Garnets enable to set up a myriad of relevant configuration and they avoid drawing conclusion on our approach from specific games.



\vspace{-0.8em}
\section{Background}
\vspace{-0.8em}
In an $N$-player MG with a finite state space $S$ all players take actions in a finite set $((A^i (s))_{s \in S} )_{i \in \{1,...,N \}}$ depending on the player and on the state. In every state $s$, each player has to take an action within his/her action space $a^i \in A^i(s)$ and receives a reward signal $r^i (s, a^1 , ..., a^N)$ which depends on the joint action of all players in state $s$. Then, we assume that the state changes according to a transition kernel $p(s' |s, a^1, ..., a^N)$.
All actions indexed by $-i$ are joint actions of all players except player $i$ (i.e. $\bm{a^{\unaryminus i}} = (a^1,\dots,a^{i-1},a^{i+1},\dots,a^N)$).
For the sake of simplicity we will note $\mathbf{a} = (a^1,\dots,a^N) = (a^i , \bm{a^{\unaryminus i}} )$, $p(s' |s, a^1, ..., a^N) = p(s' |s, \mathbf{a}) = p(s' |s, a^i , \bm{a^{\unaryminus i}} )$ and $r^i (s, a^1, ..., a^N) = r^i (s, \mathbf{a}) = r^i (s, a^i , \bm{a^{\unaryminus i}} )$. 
The constant $\gamma \in [0, 1)$ is the discount factor.

In a MARL problem, the goal is typically to find a strategy for each player (a policy in the MDP literature). A stationary strategy $\pi^i$ maps to each state a distribution over the action space $A^i(s)$. The strategy $\pi^i(.|s)$ is such a distribution. Alike previously, we will write $\bm{\pi} = (\pi^1,...,\pi^N) = (\pi^i,\bm{\pi^{\unaryminus i}})$ the product distribution of those strategies (i.e. $\bm{\pi^{\unaryminus i}} = (\pi^1,\dots,\pi^{i-1},\pi^{i+1},\dots,\pi^N)$).
The stochastic kernel defining the dynamics of the Markov Chain when all players follow their strategy $\bm{\pi}$ is $\cP_{\bm{\pi}} (s'|s) = E_{\mathbf{a} \sim \bm{\pi}}[ p(s'|s,\mathbf{a})]$ and the one defining the MDP when all players but player $i$ follow their strategy $\bm{\pi^{\unaryminus i}}$ is $\cP_{\bm{\pi^{\unaryminus i}}} (s'|s, a^i) = E_{\bm{a^{\unaryminus i}} \sim \bm{\pi^{\unaryminus i}}}[ p(s'|s,\mathbf{a})]$. The kernel $\cP_{\bm{\pi}}$ can be seen as a squared matrix of size the number of states$|S|$. 
Identically, we can define the averaged reward over all strategies $r^i_{\bm{\pi}} (s) = E_{\mathbf{a} \sim \bm{\pi}}[ r^i(s,\mathbf{a})]$ and the averaged reward over all players but player $i$, $r^i_{\bm{\pi^{\unaryminus i}}} (s, a^i) = E_{\bm{a^{\unaryminus i}} \sim \bm{\pi^{\unaryminus i}}}[ r^i(s,\mathbf{a})]$. Again, the reward $r^i_{\bm{\pi}}$ can be seen as a vector of size $|S|$. Finally, the identity matrix is noted $\cI$.

\paragraph{Value of the joint strategy :}
For each state $s$, the expected return considering each player plays his part of the joint strategy $\bm{\pi}$ is the $\gamma$-discounted sum of his rewards \cite{prasad2015two,filar2012competitive}:

\small
\begin{align*}
v^i_{\bm{\pi}} (s) &= E [\sum \limits_{t=0}^{+\infty} \gamma^t r^i_{\bm{\pi}}(s_t)|s_0 = s, s_{t+1} \sim \cP_{\bm{\pi}}(.|s_t)],\\
&= (\cI - \gamma \cP_{\bm{\pi}})^{-1}r^i_{\bm{\pi}}.
\end{align*}
\normalsize
The joint value of each player is the following: $\bm{v}_{\bm{\pi}} = (v^1_{\bm{\pi}},...,v^N_{\bm{\pi}})$. In the MG theory, two Bellman operators can be defined on the value function with respect to player $i$. The first one is $\cT^i_{\mathbf{a}} v^i \; (s) = r^i(s,\mathbf{a}) + \gamma \sum \limits_{s' \in S} p(s'|s,\mathbf{a}) v^i(s')$. It represents the expected value player $i$ will get in state $s$ if all players play the joint action $\bm{a}$ and if the value for player $i$ after one transition is $v^i$.
If we consider that instead of playing action $\bm{a}$ every player plays according to a joint strategy $\bm{\pi}$ the Bellman operator to consider is $\cT^i_{\bm{\pi}} v^i \; (s)= E_{\mathbf{a} \sim \bm{\pi}} [\cT^i_{\mathbf{a}} v^i\;(s)] = r^i_{\bm{\pi}}(s) + \gamma \sum \limits_{s' \in S} \cP_{\bm{\pi}} (s'|s) v^i(s')$. As in MDPs theory, this operator is a $\gamma$-contraction in $\cL_{+\infty}$-norm~\cite{puterman2014markov}. Thus, it admits a unique fixed point which is the value for player $i$ of the joint strategy $\bm{\pi}$ : $v^i_{\bm{\pi}}$.

\paragraph{Value of the best response :}
When the strategy $\bm{\pi^{\unaryminus i}}$ of all players but player $i$ is fixed, the problem is reduced to an MDP of kernel $\cP_{\bm{\pi^{\unaryminus i}}}$ and reward function $r^i_{\bm{\pi^{\unaryminus i}}}$. In this underlying MDP, the optimal Bellman operator would be $\cT^{*i}_{\bm{\pi^{\unaryminus i}}} v^i \;(s) = \max \limits_{a^i} E_{\bm{a^{\unaryminus i}} \sim \bm{\pi^{\unaryminus i}}}[\cT^i_{a^i,\bm{a^{\unaryminus i}}} v^i \;(s)] =  \max \limits_{a^i} [r^i_{\bm{\pi^{\unaryminus i}}} (s, a^i) + \gamma \sum \limits_{s' \in S} \cP_{\bm{\pi^{\unaryminus i}}} (s'|s, a^i) v^i(s')]$. The fixed point of this operator is the value of a best response of player $i$ when all other players play $\bm{\pi^{\unaryminus i}}$ and will be written $v^{*i}_{\bm{\pi^{\unaryminus i}}} = \max \limits_{\tilde{\pi}^i} v^i_{\tilde{\pi}^i,\bm{\pi^{\unaryminus i}}}$.
%
\vspace{-0.8em}
\section{Nash, $\epsilon$-Nash and Weak $\epsilon$-Nash Equilibrium}
\vspace{-0.8em}
A Nash equilibrium is a solution concept well defined in game theory. It states that one player cannot improve his own value by switching his strategy if the other players do not vary their own one~\cite{filar2012competitive}. The goal of this paper is to find one strategy for players which is as close as possible to a Nash equilibrium.
\begin{definition}
\label{DefNash}
In an MG, a strategy $\bm{\pi}$ is a Nash equilibrium if: 
\small
$\forall i \in \{1,...,N\}, \; v_{\bm{\pi}}^i = v^{*i}_{\bm{\pi^{\unaryminus i}}}.$
\normalsize
\end{definition}
%
%
%
This definition can be rewritten with Bellman operators:
\begin{definition}
\label{DefNashBellman}
In an MG, a strategy $\bm{\pi}$ is a Nash equilibrium if $\exists \bm{v}$ such as $\forall i \in \{1,...,N\}, \cT^i_{\bm{\pi}} v^i = v^i \textrm{ and } \cT^{*i}_{\bm{\pi^{\unaryminus i}}} v^i = v^i$.
\end{definition}
\begin{proof}
The proof is left in appendix \ref{proofEqDef}.
\end{proof}
One can notice that, in the case of a single player MG (or MDP), a Nash equilibrium is simply the optimal strategy. An $\epsilon$-Nash equilibrium is a relaxed solution concept in game theory. When all players play an $\epsilon$-Nash equilibrium the value they will receive is at most $\epsilon$ sub-optimal compared to a best response. Formally~\cite{filar2012competitive}:
\begin{definition}
In an MG, a strategy $\bm{\pi}$ is an $\epsilon$-Nash equilibrium if:\newline
$\forall i \in \{1,...,N\}, \; v_{\bm{\pi}}^i + \epsilon \geq v^{*i}_{\bm{\pi^{\unaryminus i}}}$\newline
or $\forall i \in \{1,...,N\}, \; v^{*i}_{\bm{\pi^{\unaryminus i}}} - v_{\bm{\pi}}^i \leq \epsilon,$\newline
which is equivalent to:
\small$\normp{\normp{v^{*i}_{\bm{\pi^{\unaryminus i}}} - v_{\bm{\pi}}^i}{s}{\infty}}{i}{\infty} \leq \epsilon.$\normalsize
\end{definition}
Interestingly, when considering an MDP, the definition of an $\epsilon$-Nash equilibrium is reduced to control the $\cL_{+\infty}$-norm between the value of the players' strategy and the optimal value. However, it is known that approximate dynamic programming algorithms do not control a $\cL_{+\infty}$-norm but rather an $\cL_p$-norm~\cite{munos2008finite} (we take the definition of the $\cL_p$-norm of~\cite{Piot_DC_NIPS}).
%
Using $\cL_p$ -norm is necessary for approximate dynamic programming algorithms to use complexity bounds from learning theory~\cite{Piot_DC_NIPS}.
The convergence of these algorithms was analyzed using supervised learning bounds in $\cL_p$-norm and thus guaranties are given in $\cL_p$-norm~\cite{scherrer2012approximate}. In addition, Bellman residual approaches on MDPs also give guaranties in $\cL_p$-norm~\cite{maillard2010finite,Piot_DC_NIPS}.
Thus, we define a natural relaxation of the previous definition of the $\epsilon$-Nash equilibrium in $\cL_p$-norm which is consistent with the existing work on MDPs.
\begin{definition}
In a MG, $\bm{\pi}$ is a weak $\epsilon$-Nash equilibrium if:
\small$\normp{\normp{v^{*i}_{\bm{\pi^{\unaryminus i}}} - v_{\bm{\pi}}^i}{\mu(s)}{p}}{\rho(i)}{p} \leq \epsilon.$\normalsize
\end{definition}
One should notice that an $\epsilon$-Nash equilibrium is a weak $\epsilon$-Nash equilibrium.
Conversely, a weak $\epsilon$-Nash equilibrium is not always an $\epsilon$-Nash equilibrium.
Furthermore, both $\epsilon$ do not need to be equal. 
The notion of weak $\epsilon$-Nash equilibrium defines a performance criterion to evaluate a strategy while seeking for a Nash equilibrium.
Thus, this definition has the great advantage to provide a convincing way to evaluate the final strategy. In the case of an MDP, it states that a weak $\epsilon$-Nash equilibrium only consists in controlling the difference in $\cL_p$-norm between the optimal value and the value of the learned strategy. For an MG, this criterion is an $\cL_p$-norm over players ($i \in \{1,\dots,N\}$) of an $\cL_p$-norm over states ($s \in S$) of the difference between the value of the joint strategy $\bm{\pi}$ for player $i$ in state $s$ and of his best response against the joint strategy $\bm{\pi^{\unaryminus i}}$.
In the following, we consider that learning a Nash equilibrium should result in minimizing the loss \small$\normp{\normp{v^{*i}_{\bm{\pi^{\unaryminus i}}} - v_{\bm{\pi}}^i}{\mu(s)}{p}}{\rho(i)}{p}$\normalsize. For each player, we want to minimize the difference between the value of his strategy and a best response considering the strategy of the other players is fixed. However, a direct minimization of that norm is not possible in the batch setting even for MDPs. Indeed, $v^{*i}_{\bm{\pi^{\unaryminus i}}}$ cannot be directly observed and be used as a target. A common strategy to alleviate this problem is to minimize a surrogate loss. In MDPs, a possible surrogate loss is the optimal Bellman residual $\|v-T^* v\|_{\mu,p}$ (where $\mu$ is a distribution over states and $T^*$ is the optimal Bellman operator for MDPs)~\cite{Piot_DC_NIPS,baird1995residual}. The optimal policy is then extracted from the learnt optimal value (or $Q$-value in general). In the following section, we extend this Bellman residual approach to MGs.

%
\vspace{-0.6em}
\section{Bellman Residual Minimization in MGs}
\vspace{-0.6em}
Optimal Bellman residual minimization is not straightforwardly extensible to MGs because multiplayer strategies cannot be directly extracted from value functions as shown in~\cite{zinkevich2006cyclic}. Yet, from Definition \ref{DefNashBellman}, we know that a joint strategy $\bm{\pi}$ is a Nash equilibrium if there exists $ \bm{v}$ such that, for any player $i$, $v^i$ is the value of the joint strategy $\bm{\pi}$ for player $i$ (i.e. $\cT^i_{\bm{\pi}} v^i = v^i$) and $v^i$ is the value of the best response player $i$ can achieve regarding the opponent's strategy $\bm{\pi^{\unaryminus i}}$ (i.e. $\cT^{*i}_{\bm{\pi^{\unaryminus i}}} v^i = v^i$). We thus propose to build a second Bellman-like residual optimizing over the set of strategies so as to directly learn an $\epsilon$-Nash equilibrium. The first (traditional) residual (i.e. $\normp{\cT^{*i}_{\bm{\pi^{\unaryminus i}}} v^i - v^i}{\rho(i)}{p}$) forces the value of each player to be close to their respective best response to every other player while the second residual (i.e. $\normp{\cT^i_{\bm{\pi}} v^i - v^i}{\rho(i)}{p}$) will force every player to play the strategy corresponding to that value. 

One can thus wonder how close from a Nash-Equilibrium $\bm{\pi}$ would be if there existed $\bm{v}$ such that $\cT^{*i}_{\bm{\pi^{\unaryminus i}}} v^i \approx v^i$ and $\cT^i_{\bm{\pi}} v^i \approx v^i$. In this section, we prove that, if we are able to control over $(\bm{v},\bm{\pi})$ a sum of the $\cL_p$-norm of the associated Bellman residuals (\small$\normp{\cT^{*i}_{\bm{\pi^{\unaryminus i}}} v^i - v^i}{\mu}{p}$\normalsize and \small$\normp{\cT^i_{\bm{\pi}} v^i - v^i}{\mu}{p}$\normalsize), then we are able to control \small$\normp{\normp{ v^{*i}_{\bm{\pi^{\unaryminus i}}} - v_{\bm{\pi}}^i}{\mu(s)}{p}}{\rho(i)}{p}$\normalsize.
%

\begin{theorem}
\label{theorem}
$\forall p, p'$ positive reals such that $\frac{1}{p}+\frac{1}{p'} = 1$ and $\forall (v^1,\dots,v^N)$:
\small
\begin{align}
&\normp{\normp{v^{*i}_{\bm{\pi^{\unaryminus i}}} - v_{\bm{\pi}}^i}{\mu(s)}{p}}{\rho(i)}{p} \leq \frac{2^{\frac{1}{p'}} C_{\infty}(\mu,\nu)^{\frac{1}{p}}}{1-\gamma}\\
& \quad \times \left[ \sum \limits_{i=1}^{N} \rho(i)\left(\normp{\cT^{*i}_{\bm{\pi^{\unaryminus i}}} v^i - v^i}{\nu}{p}^p+\normp{\cT^i_{\bm{\pi}} v^i - v^i}{\nu}{p}^p\right) \right]^{\frac{1}{p}},
\end{align}
\normalsize
with the following concentrability coefficient (the norm of a Radon-Nikodym derivative) \\
\small$ C_{\infty}(\mu,\nu,\pi^i,\bm{\pi^{\unaryminus i}}) = \normp{\frac{\partial \mu^T (1-\gamma)(\cI - \gamma \cP_{\pi^i,\bm{\pi^{\unaryminus i}}})^{-1}}{\partial \nu^T}}{\nu}{\infty}$\normalsize and \small$C_{\infty}(\mu,\nu) = \sup_{\bm{\pi}} C_{\infty}(\mu,\nu,\pi^i,\bm{\pi^{\unaryminus i}})$\normalsize.
\end{theorem}
\begin{proof}
The proof is left in appendix \ref{proofLemma}.
\end{proof}

This theorem shows that an $\epsilon$-Nash equilibrium can be controlled by the sum over the players of the sum of the norm of two Bellman-like residuals: the Bellman Residual of the best response of each player and the Bellman residual of the joint strategy.
If the residual of the best response of player $i$ ($\normp{\cT^{*i}_{\bm{\pi^{\unaryminus i}}} v^i - v^i}{\nu}{p}^p$) is small, then the value $v^i$ is close to the value of the best response $v^{*i}_{\bm{\pi^{\unaryminus i}}}$ and if the Bellman residual of the joint strategy $\normp{\cT^i_{\bm{\pi}} v^i - v^i}{\nu}{p}^p$ for player $i$ is small, then $v^i$ is close to $v_{\bm{\pi}}^i$. In the end, if all those residuals are small, the joint strategy is an $\epsilon$-Nash equilibrium with $\epsilon$ small since $v^{*i}_{\bm{\pi^{\unaryminus i}}} \simeq v^i \simeq v_{\bm{\pi}}^i$.

Theorem~\ref{theorem} also emphasizes the necessity of a weakened notion of an $\epsilon$-Nash equilibrium. It is much easier to control a $\cL_p$-norm than a $\cL_{\infty}$-norm with samples. In the following, the weighted sum of the norms of the two Bellman residuals will be noted as:

\noindent \small$f_{\nu,\rho,p}(\bm{\pi},\bm{v}) = \sum \limits_{i=1}^{N} \rho(i)\left(\normp{\cT^{*i}_{\bm{\pi^{\unaryminus i}}} v^i - v^i}{\nu}{p}^p+\normp{\cT^i_{\bm{\pi}} v^i - v^i}{\nu}{p}^p\right)$\normalsize.

Finding a Nash equilibrium is then reduced to a non-convex optimization problem. If we can find a $(\bm{\pi},\bm{v})$ such that $f_{\nu,\rho,p}(\bm{\pi},\bm{v}) = 0$, then the joint strategy $\bm{\pi}$ is a Nash equilibrium.
This procedure relies on a search over the joint value function space and the joint strategy space.
Besides, if the state space or the number of joint actions is too large, the search over values and strategies might be intractable. 
We addressed the issue by making use of approximate value functions and strategies.
Actually, Theorem~\ref{theorem} can be extended with function approximation. A good joint strategy $\bm{\pi}_\theta$ within an approximate strategy space $\Pi$ can still be found by computing $\bm{\pi_\theta},\bm{v_\eta} \in \argmin \limits_{\bm{\theta}, \; \bm{\eta}} f_{\nu,\rho,p}(\bm{\pi_\theta},\bm{v_\eta})$ (where $\bm{\theta}$ and $\bm{\eta}$ respectively parameterize $\bm{\pi_\theta},\bm{v_\eta}$). Even with function approximation, the learned joint strategy $\bm{\pi}_\theta$ would be at least a weak $\epsilon$-Nash equilibrium (with \small$\epsilon \leq \frac{2^{\frac{1}{p'}} C_{\infty}(\mu,\nu)^{\frac{1}{p}}}{1-\gamma} f_{\nu,\rho,p}(\bm{\pi_\theta},\bm{v_\eta})$\normalsize).
This is, to our knowledge, the first approach to solve MGs within an approximate strategy space and an approximate value function space.
\vspace{-0.5em}
\section{The Batch Scenario}
\vspace{-0.5em}
\label{Batch_section}
In this section, we explain how to learn a weak $\epsilon$-Nash equilibrium from Theorem \ref{theorem} with approximate strategies and an approximate value functions. As said previously, we will focus on the batch setting where only a finite number of historical data sampled from an MG are available.

In the batch scenario, it is common to work on state-action value functions (also named $Q$-functions). 
In MDPs, the $Q$-function is defined on the state and the action of one agent. In MGs, the $Q$-function has to be extended over the joint action of the agents~\cite{perolat,hu2003nash}.
Thus, the $Q$-function in MGs for a fixed joint strategy $\bm{\pi}$ is the expected $\gamma$-discounted sum of rewards considering that the players first pick the joint action $\bm{a}$ and follow the joint strategy $\bm{\pi}$. Formally, the $Q$-function is defined as $Q^i_{\bm{\pi}}(s,\mathbf{a}) = \cT^i_{\mathbf{a}} v^i_{\bm{\pi}}$. Moreover, one can define two analogue Bellman operators to the ones defined for the value function: 

\small
\begin{align}
\cB^i_{\bm{\pi}} Q \; (s,\mathbf{a}) &= r^i(s,\mathbf{a}) + \sum \limits_{s' \in S} p(s'|s,\mathbf{a}) E_{\mathbf{b} \sim \bm{\pi}}[Q(s',\mathbf{b})] \\
\cB^{*i}_{\bm{\pi}} Q \; (s,\mathbf{a}) &= r^i(s,\mathbf{a}) + \\
&\sum \limits_{s' \in S} p(s'|s,\mathbf{a}) \max \limits_{b^i } \left[ E_{\mathbf{\bm{b^{\unaryminus i}}} \sim \bm{\pi^{\unaryminus i}}}[Q(s',b^i, \bm{b^{\unaryminus i}})] \right].
\end{align}
\normalsize
As for the value function, $Q^i_{\bm{\pi}}$ is the fixed point of $\cB^i_{\bm{\pi}}$ and $Q^{*i}_{\bm{\pi^{\unaryminus i}}}$ is the fixed point of $\cB^{*i}_{\bm{\pi}}$ (where $Q^{*i}_{\bm{\pi^{\unaryminus i}}} = \max \limits_{\pi^i} Q^i_{\bm{\pi}}$).
The extension of Theorem \ref{theorem} to $Q$-functions instead of value functions is straightforward.
Thus, we will have to minimize the following function depending on strategies and $Q$-functions:

\small
\begin{align}
f(\mathbf{Q},\bm{\pi}) = \sum \limits_{i=1}^{N} \rho(i)\left(\normp{\cB^{*i}_{\bm{\pi^{\unaryminus i}}} Q^i - Q^i}{\nu}{p}^p+\normp{\cB^i_{\bm{\pi}} Q^i - Q^i}{\nu}{p}^p\right) \label{Q_function_batch}
\end{align}
\normalsize



The batch scenario consists in having a set of $k$ samples $(s_j,(a_j^1,...,a_j^N),(r_j^1,...,r_j^N),s'_j)_{j \in \{1,...,k\}}$ where $r_j^i = r^i(s_j,a_j^1,...,a_j^N)$ and where the next state is sampled according to $p(.|s_j,a_j^1,...,a_j^N)$.
From Equation~\eqref{Q_function_batch}, we can minimize the empirical-norm by using the $k$ samples to obtain the empirical estimator of the Bellman residual error. 
\small
\begin{align}
\tilde{f}_k(\mathbf{Q},\bm{\pi}) =\sum \limits_{j=1}^k &\sum \limits_{i=1}^{N} \rho(i) \Bigg[\abs{\cB^{*i}_{\bm{\pi^{\unaryminus i}}} Q^i (s_j,\mathbf{a}_j) - Q^i(s_j,\mathbf{a}_j)}^p\\
&+\abs{\cB^i_{\bm{\pi}} Q^i (s_j,\mathbf{a}_j) - Q^i (s_j,\mathbf{a}_j)}^p\Bigg], \label{Batch_Q_func}
\end{align}
\normalsize
For more details, an extensive analysis beyond the minimization of the Bellman residual in MDPs can be found in \cite{Piot_DC_NIPS}. In the following we discuss the estimation of the two Bellman residuals \small$\abs{\cB^{*i}_{\bm{\pi^{\unaryminus i}}} Q^i (s_j,\mathbf{a}_j) - Q^i(s_j,\mathbf{a}_j)}^p$\normalsize and \small$\abs{\cB^i_{\bm{\pi}} Q^i (s_j,\mathbf{a}_j) - Q^i (s_j,\mathbf{a}_j)}^p$\normalsize in different cases.

\paragraph{Deterministic Dynamics:}
With deterministic dynamics, the estimation is straightforward. We estimate \small$\cB^i_{\bm{\pi}} Q^i (s_j,\mathbf{a}_j)$\normalsize~with \small$r_j^i +\gamma E_{\mathbf{b} \sim \bm{\pi}}[Q^i(s'_j,\mathbf{b})]$\normalsize~and \small$\cB^{*i}_{\bm{\pi^{\unaryminus i}}} Q^i (s_j,\mathbf{a}_j)$\normalsize~with \small$r_j^i +\gamma \max \limits_{b^i } \left[ E_{\bm{b^{\unaryminus i}} \sim \bm{\pi^{\unaryminus i}}}[Q^i(s'_j,b^i, \bm{b^{\unaryminus i}})] \right]$\normalsize, where the expectation are:

\small
\begin{align*}
&E_{\mathbf{b} \sim \bm{\pi}}[Q^i(s',\mathbf{b})]=\\ 
&\sum \limits_{b^1 \in A^1} \dots \sum \limits_{b^N \in A^N} \pi^1(b^1|s') \dots \pi^N(b^N|s') Q^i(s',b^1,\dots,b^N)
\end{align*}
\normalsize and where
\small
\begin{align*}
&E_{\bm{b^{\unaryminus i}} \sim \bm{\pi^{\unaryminus i}}}[Q^i(s',b^i, \bm{b^{\unaryminus i}})] =\sum \limits_{b^1 \in A^1} \dots \sum \limits_{b^{i-1} \in A^{i-1}} \sum \limits_{b^{i+1} \in A^{i+1}} \dots \\ 
&\sum \limits_{b^N \in A^N} \pi^1(b^1|s') \dots \pi^{i-1}(b^{i-1}|s') \pi^{i+1}(b^{i+1}|s') \dots\\ 
&\quad\quad\quad\quad\quad\quad\quad\quad \times \pi^N(b^N|s') Q^i(s',b^1,\dots,b^N)
\end{align*} 
\label{eq:compute_expectation}
\normalsize
Please note that this computation can be turned into tensor operations as described in Appendix-\ref{seq:curves}.  


\paragraph{Stochastic Dynamics:}
In the case of stochastic dynamics, the previous estimator~\eqref{Q_function_batch} is known to be biased~\cite{maillard2010finite,Piot_DC_NIPS}.
If the dynamic is known, one can use the following unbiased estimator: \small$\cB^i_{\bm{\pi}} Q^i (s_j,\mathbf{a}_j)$\normalsize~is estimated with \small$r_j^i +\gamma \sum \limits_{s' \in S} p(s'|s_j,\mathbf{a}_j) E_{\mathbf{b} \sim \bm{\pi}}[Q^i(s',\mathbf{b})]$\normalsize~and \small$\cB^{*i}_{\bm{\pi^{\unaryminus i}}} Q^i (s_j,\mathbf{a}_j)$\normalsize~with \small$r_j^i +\gamma \sum \limits_{s' \in S} p(s'|s_j,\mathbf{a}_j) \max \limits_{b^i } \left[ E_{\bm{b^{\unaryminus i}} \sim \bm{\pi^{\unaryminus i}}}[Q^i(s',b^i, \bm{b^{\unaryminus i}})] \right]$\normalsize.

If the dynamic of the game is not known (e.g. batch scenario), the unbiased estimator cannot be used since the kernel of the MG is required and other techniques must be applied. 
One idea would be to first learn an approximation of this kernel. For instance, one could extend the MDP techniques to MGs such as embedding a kernel in a Reproducing Kernel Hilbert Space (RKHS)~\cite{grunewalder2012modelling,piot_BoostedBellmanResidual,Piot_DC_NIPS} or using kernel estimators~\cite{taylor2012value,Piot_DC_NIPS}.
Therefore, $p(s'|s_j,\bm{a}_j)$ would be replaced by an estimator of the dynamics $\tilde{p}(s'|s_j,\bm{a}_j)$ in the previous equation.
If a generative model is available, the issue can be addressed with double sampling as discussed in~\cite{maillard2010finite,Piot_DC_NIPS}.
%
%
\vspace{-0.5em}
\section{Neural Network architecture}
\vspace{-0.5em}
\label{sec:NNA}
Minimizing the sum of Bellman residuals is a challenging problem as the objective function is not convex. It is all the more difficult as both the $Q$-value function and the strategy of every player must be learned independently. Nevertheless, neural networks have been able to provide good solutions to problems that require minimizing non-convex objective such as image classification or speech recognition \cite{lecun2015deep}. Furthermore, neural networks were successfully applied to reinforcement learning to approximate the $Q$-function~\cite{mnih2015human} in one agent MGs with eclectic state representation.

Here, we introduce a novel neural architecture: the NashNetwork\footnote{\url{https://github.com/fstrub95/nash_network}}. For every player, a two-fold network is defined: the $Q$-network that learns a $Q$-value function and a $\pi$-network that learns the stochastic strategy of the players. The $Q$-network is a multilayer perceptron which takes the state representation as input. It outputs the predicted $Q$-values of the individual action such as the network used by \cite{mnih2015human}. Identically, the $\pi$-network is also a multilayer perceptron which takes the state representation as input. It outputs a probability distribution over the action space by using a softmax. We then compute the two Bellman residuals for every player following Equation~\ref{eq:compute_expectation}. Finally, we back-propagate the error by using classic gradient descent operations. 

In our experiments, we focus on deterministic turn-based games. It entails a specific neural architecture that is fully described in Figure~\ref{fig:network}. During the training phase, all the $Q$-networks and the $\pi$-networks of the players are used to minimize the Bellman residual. Once the training is over, only the $\pi$-network is kept to retrieve the strategy of each player. Note that this architecture differs from classic actor-critic networks~\cite{lillicrap2015continuous} for several reasons. Although $Q$-network is an intermediate support to the computation of $\pi$-network, neither policy gradient nor advantage functions are used in our model. Besides, the $Q$-network is simply discarded at the end of the training. The NashNetwork directly searches in the strategy space by minimizing the Bellman residual. 
%
%
\vspace{-0.5em}
\section{Experiments}
\vspace{-0.5em}
\label{sec:experiments}
In this section, we report an empirical evaluation of our method on randomly generated MGs. This class of problems has been first described by~\cite{archibald1995generation} for MDPs and has been studied for the minimization of the optimal Bellman residual~\cite{Piot_DC_NIPS} and in zero-sum two-player MGs~\cite{perolat_non-stat_AISTATS}. First, we extend the class of randomly generated MDPs to general-sum MGs, then we describe the training setting, finally we analyze the results. Without loss of generality we focus on deterministic turn-based MGs for practical reasons. Indeed, in simultaneous games the complexity of the state actions space grows exponentially with the number of player whereas in turn-based MGs it only grows linearly. Besides, as in the case of simultaneous actions, a Nash equilibrium in a turn-based MG (even deterministic) might be a stochastic strategy~\cite{zinkevich2006cyclic}. In a turn-based MG, only one player can choose an action in each state. Finally, we run our experiments on deterministic MGs to avoid bias in the estimator as discussed in Sec.\ref{Batch_section}.

To sum up, we use turn-based games to avoid the exponential growth of the number of actions and deterministic games to avoid bias in the estimator. The techniques described in Sec.\ref{Batch_section} could be implemented with a slight modification of the architecture described in Sec.\ref{sec:NNA}.
\begin{figure*}[ht]
\vspace{-1em}
  \centering
  \includegraphics[width=15.cm]{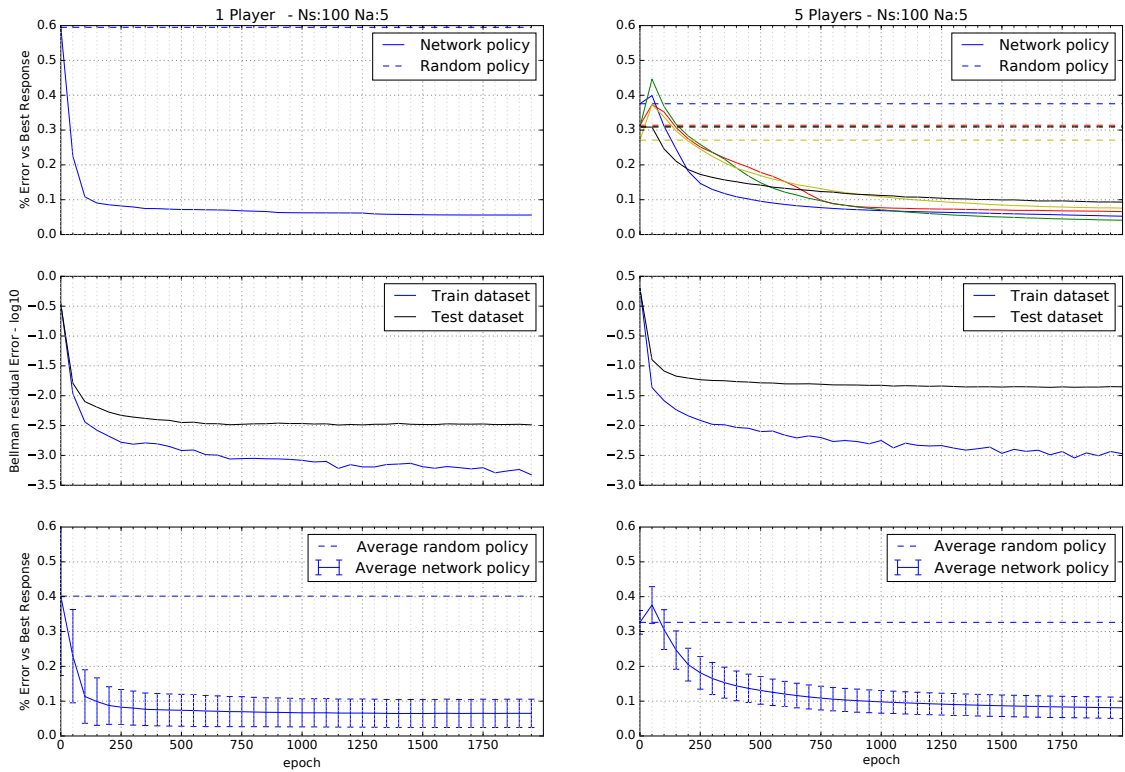}
  \caption{(top) Evolution of the Error vs Best Response during the training for a given Garnet. In both cases, players manage to iteratively improve their strategies. A strong drop may occur when one of the player finds a important move. (middle) Empirical Bellman residual for the training dataset and the testing dataset. It highlights the link between minimizing the Bellman residual and the quality of the learned strategy. (bottom) Average Error vs Best Response averaged over every players, Garnets and every batch. It highlights the robustness of minimizing the Bellman residual over several games. Experiments are run on 5 Garnets which is re-sampled 5 times to average the metrics.}
\vspace{-1em}
  \label{fig:garnet}
\end{figure*}
\vspace{-1em}
\paragraph{Dataset:}
We chose to use artificially generated MGs (named Garnet) to evaluate our method as they are not problem dependent.
They have the great advantage to allow to easily change the state space (and its underlying structure), the action space and the number of players. Furthermore, Garnets have the property to present all the characteristics of a complete MG. Standard games such as "rock-paper-scissors" (no states and zero-sum), or other games such as chess or checkers (deterministic and zero-sum), always rely on a limited number of MG properties. Using classic games would thus have hidden the generality of our approach even though it would have been more appealing. Finally, the underlying model of the dynamics of the Garnet is fully known. Thus, it is possible to investigate whether an $\epsilon$-Nash equilibrium have been reached during the training and to quantify the $\epsilon$. It would have been impossible to do so with more complex games.

An $N$-player Garnet is described by a tuple $(N_S, N_A)$. The constant $N_S$ is the number of states and $N_A$ is the number of actions of the MG. For each state and action $(s,a)$, we randomly pick the next state $s'$ in the neighborhood of indices of $s$ where $s'$ follow a rounded normal distribution $\mathcal{N}(s, \hat{\sigma}_{s'})$. We then build the transition matrix such as $p(s'|s,a) = 1$.
Next, we enforce an arbitrary reward structure into the Garnet to enable the emergence of an Nash equilibrium. Each player $i$ has a random critical state $\hat{s}^i$ and the reward of this player linearly decreases by the distance (encoded by indices) to his critical state. Therefore, the goal of the player is to get as close as possible from his critical state. Some players may have close critical states and may act together while other players have to follow opposite directions.   
A Gaussian noise of standard deviation $\hat{\sigma}_{noise}$ is added to the reward, then we sparsify it with a ratio $m_{noise}$ to harden the task.
Finally, a player is selected uniformly over $\{1,\dots,N\}$ for every state once for all. The resulting vector $v_c$ encodes which player plays at each state as it is a turn-based game.
Concretely, a Garnet works as follows: Given a state $s$, a player is first selected according the vector $v_c(s)$. This player then chooses an action according its strategy $\pi^i$. Once the action is picked, every player $i$ receives its individual reward $r^i(s,a)$. Finally, the game moves to a new state according $p(s'|s,a)$ and a new state's player ($v_c(s')$). Again, the goal of each player is to move as close as possible to its critical state. Thus, players benefit from choosing the action that maximizes its cumulative reward and leads to a state controlled by a non-adversarial player.    

Finally, samples $(s,(a^1,\dots,a^N),(r^1,\dots,r^N),s')$ are generated by randomly selecting a state and an action uniformly. The reward, the next state, and the next player are selected according to the model described above.
\vspace{-1em}
\paragraph{Evaluation:} Our goal is to verify whether the-joint strategy of the players is an $\epsilon$-Nash equilibrium. To do so, we first retrieve the strategy of every player by evaluating the $\pi^i$-networks over the state space. Then, given $\bm{\pi}$, we can exactly evaluate the value of the joint strategy $v^i_{\bm{\pi}}$ for each player $i$ and the value of the best response to the strategy of the others $v^{*i}_{\bm{\pi^{\unaryminus i}}}$. To do so, the value $v^i_{\bm{\pi}}$ is computed by inverting the linear system $v^i_{\bm{\pi}} = (\cI - \gamma \cP_{\bm{\pi}})^{-1}r^i_{\bm{\pi}}$. The value $v^{*i}_{\bm{\pi^{\unaryminus i}}}$ is computed with the policy iteration algorithm. Finally, we compute the \emph{Error vs Best Response} for every player defined as $\frac{\norm{v^i_{\bm{\pi}} - v^{*i}_{\bm{\pi^{\unaryminus i}}}}_2}{\norm{v^{*i}_{\bm{\pi^{\unaryminus i}}}}_2}$. If this metric is close to zero for all players, the players reach a weak $\epsilon$-Nash equilibrium with $\epsilon$ close to zero. Actually, \emph{Error vs Best Response} is a normalized quantification of how sub-optimal the player's strategy is compared to his best response. It indicates by which margin a player would have been able to increase his cumulative rewards by improving his strategy while other players keep playing the same strategies. 
If this metric is close to zero for all players, then they have little incentive to switch from their current strategy. It is an $\epsilon$-Nash equilibrium. In addition, we keep track of the empirical norm of the Bellman residual on both the training dataset and the test dataset as it is our training objective.
\vspace{-1em}
\paragraph{Training parameters:}
We use $N$-player Garnet with 1, 2 or 5 players. The state space and the action space are respectively of size 100 and 5. The state is encoded by a binary vector. The transition kernel is built with a standard deviation $\hat{\sigma}_{s'}$ of 1. 
The reward function is $r^i(s)= \frac{2\min(|s-\hat{s}^i|, N_S)-|s-\hat{s}^i|)}{N_S}$ (it is a circular reward function). The reward sparsity $m_{noise}$ is set to 0.5 and the reward white noise $\hat{\sigma}_{noise}$ has a standard deviation of 0.05. The discount factor $\gamma$ is set to 0.9. The $Q$-networks and $\pi$-networks have one hidden layers of size 80 with RELUs~\cite{Goodfellow-et-al-2016-Book}. $Q$-networks have no output transfer function while $\pi$-networks have a softmax. The gradient descent is performed with AdamGrad~\cite{Goodfellow-et-al-2016-Book} with an initial learning rate of 1e-3 for the $Q$-network and 5e-5 for the $\pi$-networks. We use a weight decay of 1e-6. The training set is composed of $5 N_S N_A$ samples split into random minibatch of size 20 while the testing dataset contains $N_S N_A$ samples. 
The neural network is implemented by using the python framework Tensorflow~\cite{tensorflow2015-whitepaper}. 
The source code is available on Github (Hidden for blind review) to run the experiments.
\vspace{-1em}
\paragraph{Results:}
Results for 1 player (MDP) and 5 players are reported in Figure \ref{fig:garnet}. Additional settings are reported in the Appendix-\ref{seq:curves} such as the two-player case and the tabular cases. In those scenarios, the quality of the learned strategies converges in parallel with the empirical Bellman residual. Once the training is over, the players can increase their cumulative reward by no more than 8\% on average over the state space. Therefore, neural networks succeed in learning a weak $\epsilon$-Nash equilibrium. Note that it is impossible to reach a zero error as (i) we are in the batch setting, (ii) we use function approximation and (iii) we only control a weak $\epsilon$-Nash equilibrium. Moreover, the quality of the strategies are well-balanced among the players as the standard deviation is below 5 points. 

Our neural architecture has good scaling properties. First, scaling from 2 players to 5 results in the same strategy quality. Furthermore, it can be adapted to a wide variety of problems by only changing the bottom of each network to fit with the state representation of the problem.
\vspace{-1em}
\paragraph{Discussions:} 
The neural architecture faces some over-fitting issues. It requires a high number of samples to converge as described on Figure~\ref{fig:sampling}. \cite{lillicrap2015continuous} introduces several tricks that may improve the training. Furthermore, we run additional experiments on non-deterministic Garnets but the result remains less conclusive. Indeed, the estimators of the Bellman residuals are biased for stochastic dynamics. As discussed, embedding the kernel or using kernel estimators may help to estimate properly the cost function (Equation~\eqref{Batch_Q_func}).
Finally, our algorithm only seeks for a single Nash-Equilibrium when several equilibria might exist. Finding a specific equilibrium among others is out of the scope of this paper.  
\begin{figure}
  \centering
  \includegraphics[width=0.90\columnwidth]{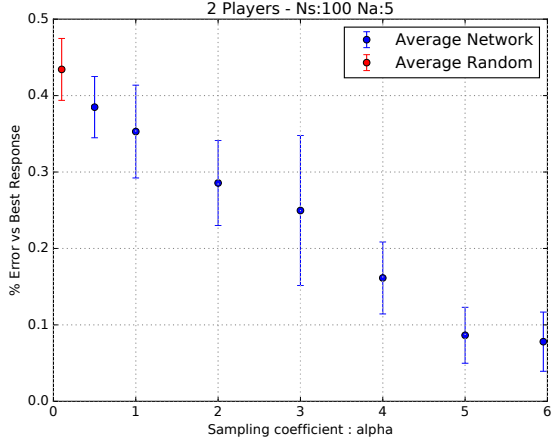}
  \caption{Impact of the number of samples on the quality of the learned strategy. The number of samples per batch is computed by $N_{samples}=\alpha N_A N_S$. Experiments are run on 3 different Garnets which are re-sampled 3 times to average the metrics.}
  \label{fig:sampling}
  \vskip -0.5em
\end{figure}

\vspace{-1.5em}
\section{Conclusion}
\vspace{-0.5em}
In this paper, we present a novel approach to learn a Nash equilibrium in MGs. The contributions of this paper are both theoretical and empirical. First, we define a new (weaker) concept of an $\epsilon$-Nash equilibrium. 
We prove that minimizing the sum of different Bellman residuals is sufficient to learn a weak $\epsilon$-Nash equilibrium. From this result, we provide empirical estimators of these Bellman residuals from batch data. Finally, we describe a novel neural network architecture, called NashNetwork, to learn a Nash equilibrium from batch data. This architecture does not rely on a specific MG. It also scales to a high number of players. Thus it can be applied to a trove of applications.

As future works, the NashNetwork could be extended to more eclectic games such as simultaneous games (i.e. Alesia~\cite{perolat} or a stick together game such as in~\cite{prasad2015two}) or Atari's games with several players such as Pong or Bomber. Additional optimization methods can be studied to this specific class of neural networks to increase the quality of learning. 
Moreover, the $Q$-function and the strategy could be parametrised with other classes of function approximation such as trees. Yet, it requires to study functional gradient descent on the loss function. And finally, future work could explore the use of projected Bellman residual as studied in MDPs and two-player zero-sum MGs~\cite{perolat:Softened}

\section*{Acknowledgments}
This work has received partial funding from CPER Nord-Pas de Calais/FEDER DATA Advanced data science and technologies 2015-2020, the European Commission H2020 framework program under the research Grant number 687831 (BabyRobot) and the Project CHIST-ERA IGLU.


\bibliographystyle{apalike}
\bibliography{biblio}

\begin{thebibliography}{}

\bibitem[Abadi et~al., 2015]{tensorflow2015-whitepaper}
Abadi, M., Agarwal, A., Barham, P., Brevdo, E., Chen, Z., Citro, C., Corrado,
  G.~S., Davis, A., Dean, J., Devin, M., Ghemawat, S., Goodfellow, I., Harp,
  A., Irving, G., Isard, M., Jia, Y., Jozefowicz, R., Kaiser, L., Kudlur, M.,
  Levenberg, J., Man\'{e}, D., Monga, R., Moore, S., Murray, D., Olah, C.,
  Schuster, M., Shlens, J., Steiner, B., Sutskever, I., Talwar, K., Tucker, P.,
  Vanhoucke, V., Vasudevan, V., Vi\'{e}gas, F., Vinyals, O., Warden, P.,
  Wattenberg, M., Wicke, M., Yu, Y., and Zheng, X. (2015).
\newblock {TensorFlow}: Large-scale machine learning on heterogeneous systems.
\newblock Software available from tensorflow.org.

\bibitem[Akchurina, 2010]{akchurina2010multi}
Akchurina, N. (2010).
\newblock {\em {Multi-Agent Reinforcement Learning Algorithms}}.
\newblock PhD thesis, Paderborn, Univ., Diss., 2010.

\bibitem[Archibald et~al., 1995]{archibald1995generation}
Archibald, T., McKinnon, K., and Thomas, L. (1995).
\newblock {On the Generation of Markov Decision Processes}.
\newblock {\em Journal of the Operational Research Society}, 46:354--361.

\bibitem[Baird et~al., 1995]{baird1995residual}
Baird, L. et~al. (1995).
\newblock {Residual Algorithms: Reinforcement Learning with Function
  Approximation}.
\newblock In {\em Proc. of ICML}.

\bibitem[Filar and Vrieze, 2012]{filar2012competitive}
Filar, J. and Vrieze, K. (2012).
\newblock {\em {Competitive Markov Decision Processes}}.
\newblock Springer Science \& Business Media.

\bibitem[Goodfellow et~al., 2016]{Goodfellow-et-al-2016-Book}
Goodfellow, I., Bengio, Y., and Courville, A. (2016).
\newblock {Deep Learning}.
\newblock Book in preparation for MIT Press.

\bibitem[Grunewalder et~al., 2012]{grunewalder2012modelling}
Grunewalder, S., Lever, G., Baldassarre, L., Pontil, M., and Gretton, A.
  (2012).
\newblock {Modelling Transition Dynamics in MDPs With RKHS Embeddings}.
\newblock In {\em Proc. of ICML}.

\bibitem[Hu and Wellman, 2003]{hu2003nash}
Hu, J. and Wellman, M.~P. (2003).
\newblock {Nash Q-Learning for General-Sum Stochastic Games}.
\newblock {\em Journal of Machine Learning Research}, 4:1039--1069.

\bibitem[Lagoudakis and Parr, 2002]{lagoudakis2002value}
Lagoudakis, M.~G. and Parr, R. (2002).
\newblock {Value Function Approximation in Zero-Sum Markov Games}.
\newblock In {\em Proc. of UAI}.

\bibitem[Lagoudakis and Parr, 2003]{lagoudakis2003reinforcement}
Lagoudakis, M.~G. and Parr, R. (2003).
\newblock {Reinforcement Learning as Classification: Leveraging Modern
  Classifiers}.
\newblock In {\em Proc. of ICML}.

\bibitem[LeCun et~al., 2015]{lecun2015deep}
LeCun, Y., Bengio, Y., and Hinton, G. (2015).
\newblock {Deep Learning}.
\newblock {\em Nature}, 521:436--444.

\bibitem[Lillicrap et~al., 2016]{lillicrap2015continuous}
Lillicrap, T.~P., Hunt, J.~J., Pritzel, A., Heess, N., Erez, T., Tassa, Y.,
  Silver, D., and Wierstra, D. (2016).
\newblock {Continuous Control with Deep Reinforcement Learning}.
\newblock In {\em Proc. of ICLR}.

\bibitem[Littman, 1994]{littman1994markov}
Littman, M.~L. (1994).
\newblock {Markov Games as a Framework for Multi-Agent Reinforcement Learning}.
\newblock In {\em Proc. of ICML}.

\bibitem[Littman, 2001]{littman2001friend}
Littman, M.~L. (2001).
\newblock {Friend-or-Foe Q-Learning in General-Sum Games}.
\newblock In {\em Proc. of ICML}.

\bibitem[Maillard et~al., 2010]{maillard2010finite}
Maillard, O.-A., Munos, R., Lazaric, A., and Ghavamzadeh, M. (2010).
\newblock {Finite-Sample Analysis of Bellman Residual Minimization}.
\newblock In {\em Proc. of ACML}.

\bibitem[Mnih et~al., 2015]{mnih2015human}
Mnih, V., Kavukcuoglu, K., Silver, D., Rusu, A.~A., Veness, J., Bellemare,
  M.~G., Graves, A., Riedmiller, M., Fidjeland, A.~K., Ostrovski, G., et~al.
  (2015).
\newblock {Human-Level Control Through Deep Reinforcement Learning}.
\newblock {\em Nature}, 518:529--533.

\bibitem[Munos and Szepesv{\'a}ri, 2008]{munos2008finite}
Munos, R. and Szepesv{\'a}ri, C. (2008).
\newblock {Finite-Time Bounds for Fitted Value Iteration}.
\newblock {\em The Journal of Machine Learning Research}, 9:815--857.

\bibitem[Nisan et~al., 2007]{nisan2007algorithmic}
Nisan, N., Roughgarden, T., Tardos, E., and Vazirani, V.~V. (2007).
\newblock {\em {Algorithmic Game Theory}}, volume~1.
\newblock Cambridge University Press Cambridge.

\bibitem[P{\'e}rolat et~al., 2016]{perolat:Softened}
P{\'e}rolat, J., Piot, B., Geist, M., Scherrer, B., and Pietquin, O. (2016).
\newblock {Softened Approximate Policy Iteration for Markov Games}.
\newblock In {\em Proc. of ICML}.

\bibitem[Perolat et~al., 2016]{perolat_non-stat_AISTATS}
Perolat, J., Piot, B., Scherrer, B., and Pietquin, O. (2016).
\newblock On the use of non-stationary strategies for solving two-player
  zero-sum markov games.
\newblock In {\em Proc. of AISTATS}.

\bibitem[Perolat et~al., 2015]{perolat}
Perolat, J., Scherrer, B., Piot, B., and Pietquin, O. (2015).
\newblock {Approximate Dynamic Programming for Two-Player Zero-Sum Markov
  Games}.
\newblock In {\em Proc. of ICML}.

\bibitem[Piot et~al., 2014a]{piot_BoostedBellmanResidual}
Piot, B., Geist, M., and Pietquin, O. (2014a).
\newblock {Boosted Bellman Residual Minimization Handling Expert
  Demonstrations}.
\newblock In {\em Proc. of ECML}.

\bibitem[Piot et~al., 2014b]{Piot_DC_NIPS}
Piot, B., Geist, M., and Pietquin, O. (2014b).
\newblock Difference of convex functions programming for reinforcement
  learning.
\newblock In {\em Proc. of NIPS}.

\bibitem[Prasad et~al., 2015]{prasad2015two}
Prasad, H., LA, P., and Bhatnagar, S. (2015).
\newblock {Two-Timescale Algorithms for Learning Nash Equilibria in General-Sum
  Stochastic Games}.
\newblock In {\em Proc. of AAMAS}.

\bibitem[Puterman, 1994]{puterman2014markov}
Puterman, M.~L. (1994).
\newblock {\em {Markov Decision Processes: Discrete Stochastic Dynamic
  Programming}}.
\newblock John Wiley \& Sons.

\bibitem[Riedmiller, 2005]{riedmiller2005neural}
Riedmiller, M. (2005).
\newblock {Neural Fitted Q Iteration--First Experiences with a Data Efficient
  Neural Reinforcement Learning Method}.
\newblock In {\em Proc. of ECML}.

\bibitem[Scherrer et~al., 2012]{scherrer2012approximate}
Scherrer, B., Ghavamzadeh, M., Gabillon, V., and Geist, M. (2012).
\newblock {Approximate Modified Policy Iteration}.
\newblock In {\em {Proc. of ICML}}.

\bibitem[Shapley, 1953]{shapley1953stochastic}
Shapley, L.~S. (1953).
\newblock {Stochastic Games}.
\newblock {\em In Proc. of the National Academy of Sciences of the United
  States of America}.

\bibitem[Taylor and Parr, 2012]{taylor2012value}
Taylor, G. and Parr, R. (2012).
\newblock {Value Function Approximation in Noisy Environments Using Locally
  Smoothed Regularized Approximate Linear Programs}.
\newblock In {\em Proc. of UAI}.

\bibitem[Zinkevich et~al., 2006]{zinkevich2006cyclic}
Zinkevich, M., Greenwald, A., and Littman, M. (2006).
\newblock {Cyclic Equilibria in Markov Games}.
\newblock In {\em Proc. of NIPS}.

\end{thebibliography}

\appendix
\onecolumn
\section{Proof of the equivalence of definition \ref{DefNash} and \ref{DefNashBellman}}
\label{proofEqDef}
Proof of the equivalence between definition \ref{DefNashBellman} and \ref{DefNash}.

(\ref{DefNashBellman}) $\Rightarrow$ (\ref{DefNash}):

If $\exists \bm{v}$ such as $\forall i \in \{1,...,N\}, \cT^i_{\bm{\pi}} v^i = v^i \textrm{ and } \cT^{*i}_{\bm{\pi^{\unaryminus i}}} v^i = v^i$, then $\forall i \in \{1,...,N\}$, $v^i = v_{\pi^i,\bm{\pi^{\unaryminus i}}}^i$ and $v^i = \max \limits_{\tilde{\pi^i}} v^i_{\tilde{\pi}^i,\bm{\pi^{\unaryminus i}}}$

(\ref{DefNash}) $\Rightarrow$ (\ref{DefNashBellman}):

if $\forall i \in \{1,...,N\}, \; v_{\pi^i,\bm{\pi^{\unaryminus i}}}^i = \max \limits_{\tilde{\pi}^i} v^i_{\tilde{\pi}^i,\bm{\pi^{\unaryminus i}}}.$, then $\forall i \in \{1,...,N\}$, the value $v^i = v_{\pi^i,\bm{\pi^{\unaryminus i}}}^i$ is such as $\cT^i_{\bm{\pi}} v^i = v^i \textrm{ and } \cT^{*i}_{\bm{\pi^{\unaryminus i}}} v^i = v^i$






\section{Proof of Theorem \ref{theorem}}
\label{proofLemma}

First we will prove the following lemma. The proof is strongly inspired by previous work on the minimization of the Bellman residual for MDPs~\cite{Piot_DC_NIPS}.
\begin{lemma}
let $p$ and $p'$ be a real numbers such that $\frac{1}{p}+\frac{1}{p'} = 1$, then $\forall \mathbf{v}, \bm{\pi}$ and $\forall i \in \{1,...,N\}$:
\small
\begin{align}
&\normp{v^i_{\pi_*^i, \bm{\pi^{\unaryminus i}}} - v^i_{\pi^i, \bm{\pi^{\unaryminus i}}}}{\mu}{p}\\
& \; \leq \frac{1}{1-\gamma} \left( C_{\infty}(\mu,\nu,\pi_*^i,\bm{\pi^{\unaryminus i}})^{\frac{p'}{p}} + C_{\infty}(\mu,\nu,\pi^i,\bm{\pi^{\unaryminus i}})^{\frac{p'}{p}}\right)^{\frac{1}{p'}} \left[\normp{\cT^{*i}_{\bm{\pi^{\unaryminus i}}} v^i - v^i}{\mu}{p}^p+\normp{\cT^i_{\bm{\pi}} v^i - v^i}{\mu}{p}^p\right]^{\frac{1}{p}},
\end{align}
\normalsize
where $\pi_*^i$ is the best response to $\bm{\pi^{\unaryminus i}}$. Meaning $v^i_{\pi_*^i, \bm{\pi^{\unaryminus i}}}$ is the fixed point of $\cT^{*i}_{\bm{\pi^{\unaryminus i}}}$. And with the following concentrability coefficient \small$ C_{\infty}(\mu,\nu,\pi^i,\bm{\pi^{\unaryminus i}}) = \normp{\frac{\partial \mu^T (1-\gamma)(\cI - \gamma \cP_{\pi^i,\bm{\pi^{\unaryminus i}}})^{-1}}{\partial \nu^T}}{\nu}{\infty}$\normalsize.
\end{lemma}
\begin{proof}
The proof uses similar techniques as in \cite{Piot_DC_NIPS}. First we have:
\begin{align}
v^i_{\pi^i, \bm{\pi^{\unaryminus i}}} - v^i &= (\cI - \gamma \cP_{\pi^i,\bm{\pi^{\unaryminus i}}})^{-1}(r^i_{\pi^i, \bm{\pi^{\unaryminus i}}} - (\cI - \gamma \cP_{\pi^i,\bm{\pi^{\unaryminus i}}}) v^i),\\
&=(\cI - \gamma \cP_{\pi^i,\bm{\pi^{\unaryminus i}}})^{-1} (\cT^i_{\pi^i,\bm{\pi^{\unaryminus i}}} v^i - v^i).
\end{align}
But we also have:
$$v^i_{\pi_*^i, \bm{\pi^{\unaryminus i}}} - v^i = (\cI - \gamma \cP_{\pi_*^i,\bm{\pi^{\unaryminus i}}})^{-1} (\cT^i_{\pi_*^i,\bm{\pi^{\unaryminus i}}} v^i - v^i),$$
then:
\begin{align}
v^i_{\pi_*^i, \bm{\pi^{\unaryminus i}}} - v^i_{\pi^i, \bm{\pi^{\unaryminus i}}} &= v^i_{\pi_*^i, \bm{\pi^{\unaryminus i}}} - v^i + v^i - v^i_{\pi^i, \bm{\pi^{\unaryminus i}}},\\
&= (\cI - \gamma \cP_{\pi_*^i,\bm{\pi^{\unaryminus i}}})^{-1} (\cT^i_{\pi_*^i,\bm{\pi^{\unaryminus i}}} v^i - v^i) - (\cI - \gamma \cP_{\pi^i,\bm{\pi^{\unaryminus i}}})^{-1} (\cT^i_{\pi^i,\bm{\pi^{\unaryminus i}}} v^i - v^i),\\
&\leq (\cI - \gamma \cP_{\pi_*^i,\bm{\pi^{\unaryminus i}}})^{-1} (\cT^{*i}_{\bm{\pi^{\unaryminus i}}} v^i - v^i) - (\cI - \gamma \cP_{\pi^i,\bm{\pi^{\unaryminus i}}})^{-1} (\cT^i_{\pi^i,\bm{\pi^{\unaryminus i}}} v^i - v^i),\\
&\leq (\cI - \gamma \cP_{\pi_*^i,\bm{\pi^{\unaryminus i}}})^{-1} \abs{\cT^{*i}_{\bm{\pi^{\unaryminus i}}} v^i - v^i} + (\cI - \gamma \cP_{\pi^i,\bm{\pi^{\unaryminus i}}})^{-1} \abs{\cT^i_{\pi^i,\bm{\pi^{\unaryminus i}}} v^i - v^i}.
\end{align}
Finally, using the same technique as the one in \cite{Piot_DC_NIPS}, we get:
\begin{align}
&\normp{v^i_{\pi_*^i, \bm{\pi^{\unaryminus i}}} - v^i_{\pi^i, \bm{\pi^{\unaryminus i}}}}{\mu}{p}\\
&\leq \normp{(\cI - \gamma \cP_{\pi_*^i,\bm{\pi^{\unaryminus i}}})^{-1} \abs{\cT^{*i}_{\bm{\pi^{\unaryminus i}}} v^i - v^i}}{\mu}{p} + \normp{(\cI - \gamma \cP_{\pi^i,\bm{\pi^{\unaryminus i}}})^{-1} \abs{\cT^i_{\pi^i,\bm{\pi^{\unaryminus i}}} v^i - v^i}}{\mu}{p},\\
&\leq \frac{1}{1-\gamma} \left [ C_{\infty}(\mu,\nu,\pi_*^i,\bm{\pi^{\unaryminus i}})^{\frac{1}{p}} \normp{\cT^{*i}_{\bm{\pi^{\unaryminus i}}} v^i - v^i}{\nu}{p} + C_{\infty}(\mu,\nu,\pi^i,\bm{\pi^{\unaryminus i}})^{\frac{1}{p}} \normp{\cT^{*i}_{\bm{\pi^{\unaryminus i}}} v^i - v^i}{\nu}{p}\right],\\
& \leq \frac{1}{1-\gamma} \left( C_{\infty}(\mu,\nu,\pi_*^i,\bm{\pi^{\unaryminus i}})^{\frac{p'}{p}} + C_{\infty}(\mu,\nu,\pi^i,\bm{\pi^{\unaryminus i}})^{\frac{p'}{p}}\right)^{\frac{1}{p'}} \left[\normp{\cT^{*i}_{\bm{\pi^{\unaryminus i}}} v^i - v^i}{\nu}{p}^p+\normp{\cT^i_{\bm{\pi}} v^i - v^i}{\nu}{p}^p\right]^{\frac{1}{p}}.
\end{align}
\end{proof}
Theorem~\ref{theorem} falls in two steps:
\small
\begin{align}
\normp{\normp{\max_{\tilde{\pi}^i} v_{\tilde{\pi}^i,\bm{\pi^{\unaryminus i}}} - v_{\pi}^i}{\mu(s)}{p}}{\rho(i)}{p}
&\leq \frac{1}{1-\gamma} \left[ \max \limits_{i \in \{1,...,N\}} \left( C_{\infty}(\mu,\nu,\pi_*^i,\bm{\pi^{\unaryminus i}})^{\frac{p'}{p}} + C_{\infty}(\mu,\nu,\pi^i,\bm{\pi^{\unaryminus i}})^{\frac{p'}{p}}\right)^{\frac{1}{p'}} \right]\\ 
& \qquad \times \left[ \sum \limits_{i=1}^{N} \rho(i)\left(\normp{\cT^{*i}_{\bm{\pi^{\unaryminus i}}} v^i - v^i}{\nu}{p}^p+\normp{\cT^i_{\bm{\pi}} v^i - v^i}{\nu}{p}^p\right) \right]^{\frac{1}{p}},\\
&\leq \frac{2^{\frac{1}{p'}} C_{\infty}(\mu,\nu)^{\frac{1}{p}}}{1-\gamma} \left[ \sum \limits_{i=1}^{N} \rho(i)\left(\normp{\cT^{*i}_{\bm{\pi^{\unaryminus i}}} v^i - v^i}{\nu}{p}^p+\normp{\cT^i_{\bm{\pi}} v^i - v^i}{\nu}{p}^p\right) \right]^{\frac{1}{p}},
\end{align}
\normalsize
with \small$C_{\infty}(\mu,\nu) = \left( \sup \limits_{\pi^i,\bm{\pi^{\unaryminus i}}} C_{\infty}(\mu,\nu,\pi^i,\bm{\pi^{\unaryminus i}}) \right)$\normalsize

The first inequality is proven using lemma 1 and Holder inequality. The second inequality falls noticing $\forall \pi^i, \bm{\pi^{\unaryminus i}}, \; C_{\infty}(\mu,\nu,\pi^i,\bm{\pi^{\unaryminus i}}) \leq \sup \limits_{\pi^i,\bm{\pi^{\unaryminus i}}} C_{\infty}(\mu,\nu,\pi^i,\bm{\pi^{\unaryminus i}})$.

\section{Additional curves}
\label{seq:curves}

This section provides additional curves regarding the training of the NashNetwork. 

\begin{figure}[ht]
  \centering
  \includegraphics[width=0.5\textwidth]{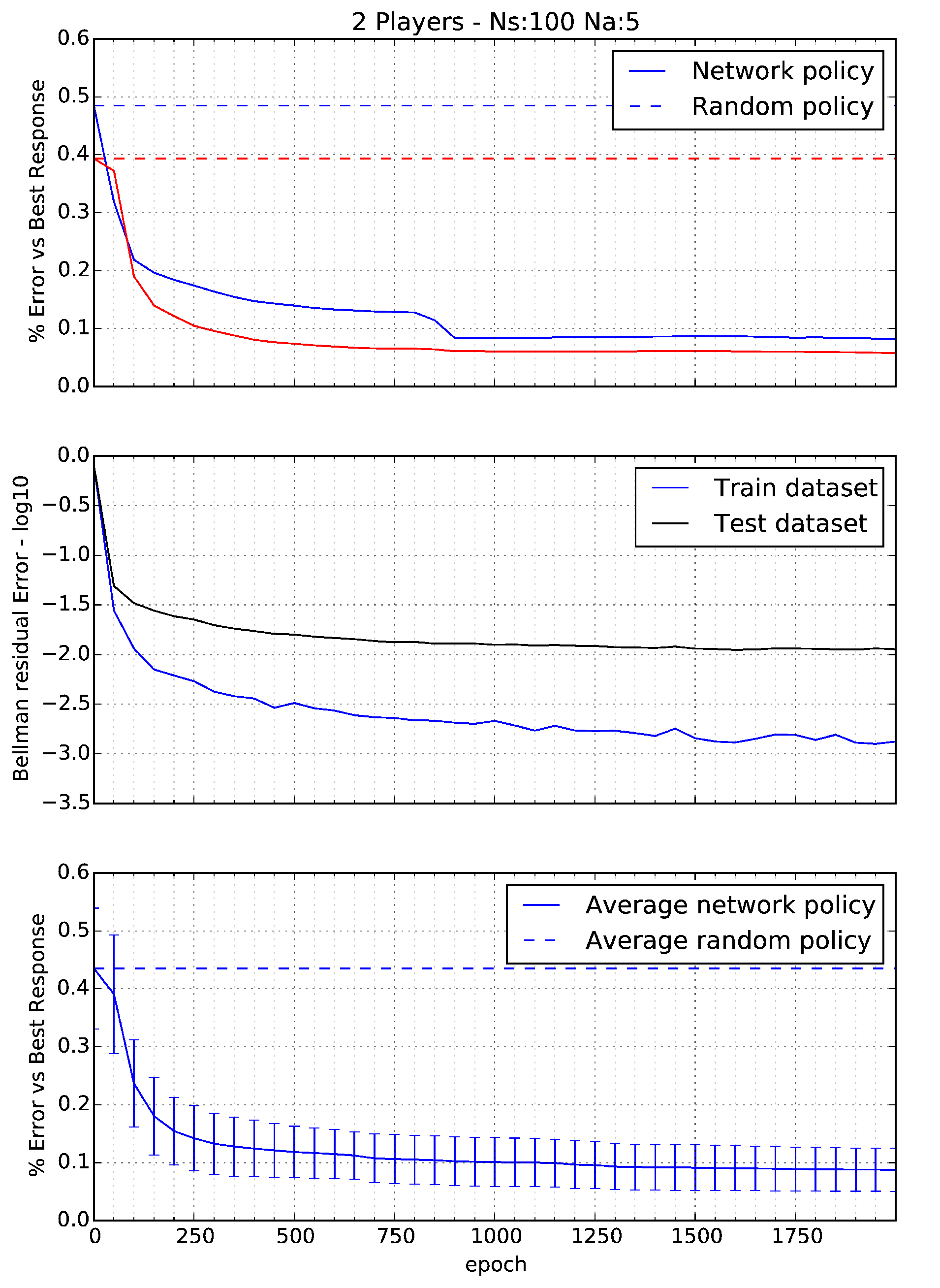}
  \caption{(top) Evolution of the Error vs Best Response during the training for a given Garnet. (middle) Empirical Bellman residual for the training dataset and the testing dataset.  (bottom) Average Error vs Best Response averaged over every players, Garnets and every batch. Experiments are run on 5 Garnets which is re-sampled 5 times to average the metrics.}
  \label{fig:garnet2}
\end{figure}

\begin{figure}[ht]
  \centering
  \includegraphics[width=1\textwidth]{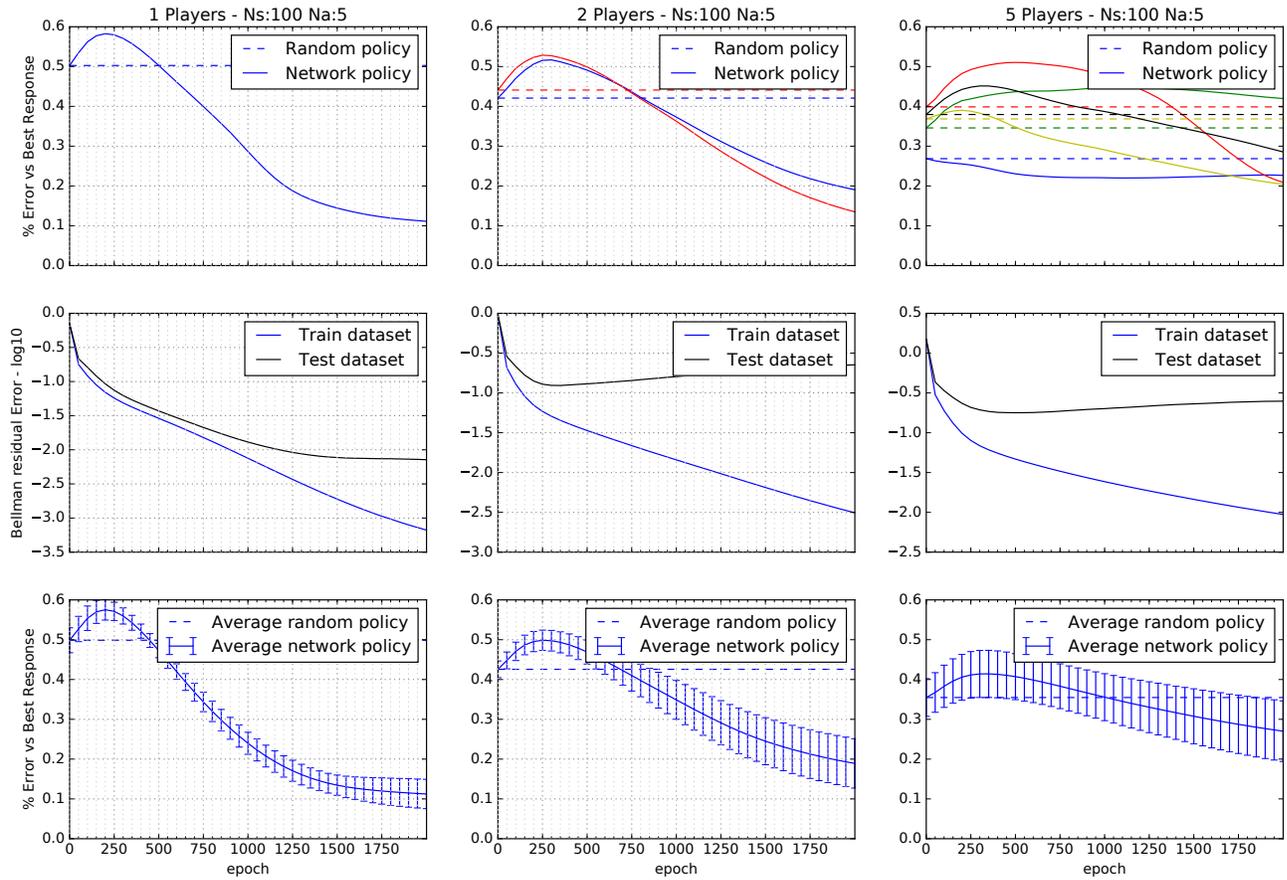}
  \caption{Tabular basis. One may notice that the tabular case works well for a 1 player game (MDP). Yet, the more players there are, the worth it performs. Experiments are run on 5 Garnets which is re-sampled 5 times to average the metrics.}
  \label{fig:garnet_tab}
\end{figure}

\begin{figure}[ht]
  \centering
  \includegraphics[width=0.8\textwidth]{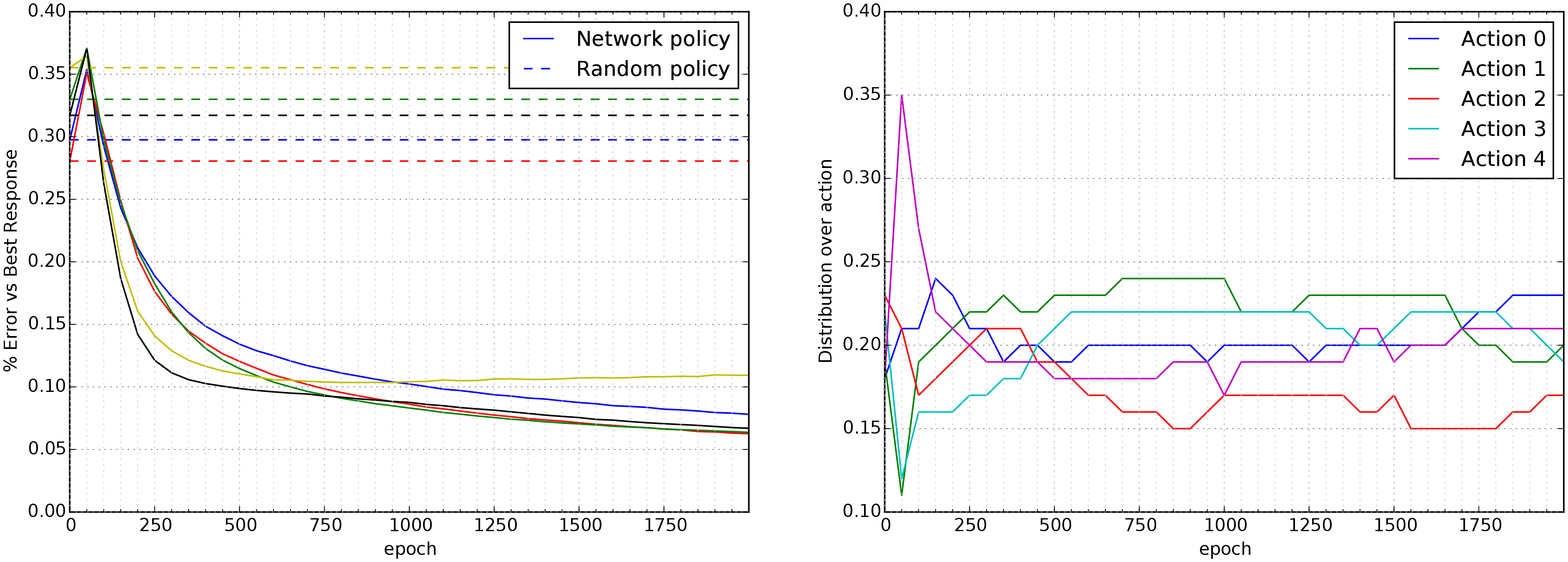}
  \caption{(left)Evolution of the Error vs Best Response during the training for a given Garnet. When the Garnet has a complex structure or the batch is badly distributed, one player sometimes fails to learn a good strategy. (right) Distribution of actions in the strategy among the players with the highest probability. This plots highlights that $\pi$-networks do modify the strategy during the training}
    \label{fig:action}
\end{figure}

\begin{figure}[ht]
  \centering
  \includegraphics[width=0.8\textwidth]{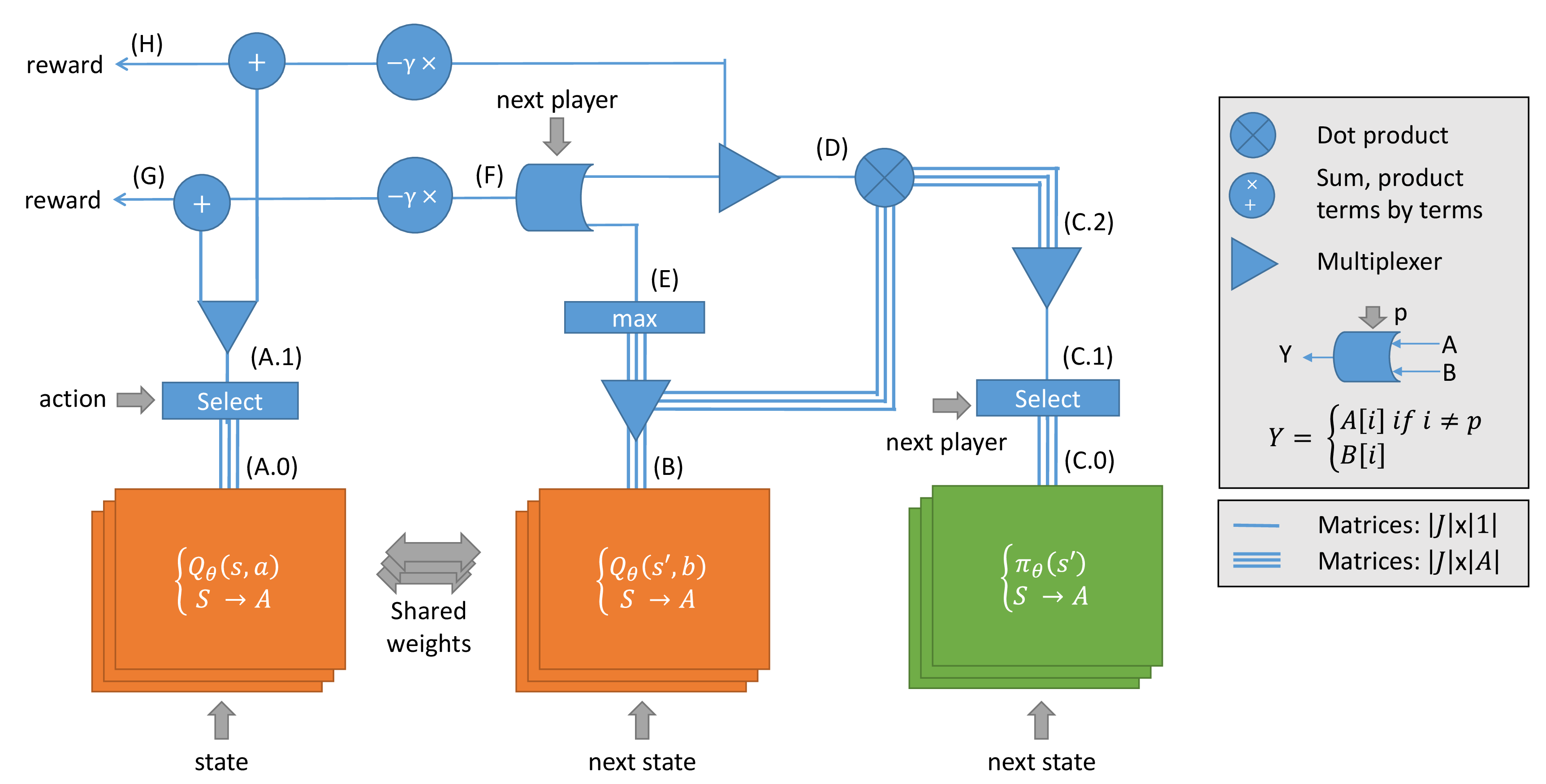}
  \caption{NashNetwork. This scheme summarizes the neural architecture that learns a Nash equilibrium by minimizing the Bellman residual for turn-based games. Each player has two networks : a $Q$-network that learns a state-action value function and a $\pi$-network that learns the strategy of the player. The final loss is an empirical estimate of the sum of Bellman residuals given a batch of data of the following shape $(s,(a^1,\dots,a^N),(r^1,\dots,r^N),s')$. The equation~\eqref{Batch_Q_func} can be divided into key operations that are described below. For each player $i$ : 
first of all, in {\bf(A.0)} we compute $Q^i(s,a^1,\dots,a^N)$ by computing the tensor $Q^i (s, \bm{a})$ and then by selecting the value of $Q^i (s, \bm{a})$ {\bf(A.1)} given the action of the batch. Step {\bf(B)} computes the tensor $Q^i (s', \bm{b})$ and step {\bf(C.0)} computes the strategy $\pi^i (.|s')$. In all Bellman residuals we need to average over the strategy of players $Q^i (s', \bm{b})$. Since we focus on turn-based MGs, we will only average over the strategy of the player controlling next state $s'$ (in the following, this player is called the next player). In {\bf(C.1)} we select the strategy $\pi(.|s')$ of the next player given the batch. In {\bf(C.2)} we duplicate the strategy of the next player for all other players. In {\bf(D)} we compute the dot product between the $Q^i (s', \bm{b})$ and the strategy of the next player to obtain $E_{\mathbf{b} \sim \bm{\pi}}[Q^i(s',\mathbf{b})]$ and in {\bf(E)} we pick the highest expected rewards and obtain $\max \limits_{\bm{b}} Q^i(s',\bm{b})$. Step {\bf (F)} aims at computing $\max \limits_{b^i } \left[ E_{\bm{b^{\unaryminus i}} \sim \bm{\pi^{\unaryminus i}}}[Q^i(s',b^i, \bm{b^{\unaryminus i}})] \right]$ and, since we deal with a turn-based MG, we need to select between the output of {\bf (D)} or {\bf (E)} according to the next player. For all $i$, we either select the one coming from {\bf (E)} if the next player is $i$ or the one from {\bf (D)} otherwise. In {\bf(G)} we compute the error between the reward $r^i$ to $Q^i(s,\mathbf{a}) - \gamma \max \limits_{b^i } \left[ E_{\bm{b^{\unaryminus i}} \sim \bm{\pi^{\unaryminus i}}}[Q^i(s',b^i, \bm{b^{\unaryminus i}})] \right]$ and in {\bf(H)} between $r^i$ and $\gamma E_{\mathbf{b} \sim \bm{\pi}}[Q^i(s',\mathbf{b})] - Q^i(s,\mathbf{a})$. The final loss is the sum of all the residuals.}
\label{fig:network}
\end{figure}
\end{document}